\documentclass[runningheads]{llncs}
\usepackage[T1]{fontenc}
\usepackage{graphicx}
\usepackage{natbib}
\usepackage{amsmath}
\usepackage[dvipsnames]{xcolor}
\usepackage{tcolorbox}
\usepackage{amssymb}
\usepackage{hyperref}
\usepackage{xurl}
\usepackage{enumitem}
\usepackage[utf8]{inputenc}

\begin{document}
\title{Why Open Source? A Game-Theoretic Analysis of the AI Race}
%
%
\author{
Andjela Mladenovic\inst{1,2} \and
Aaron Courville\inst{1,2,3} \and
Gauthier Gidel\inst{1,2,3}
}
\authorrunning{Mladenovic et al.}

\institute{
Université de Montréal \and
Mila -- Quebec AI Institute \and
Canada CIFAR AI Chair
}
%
%
%
\maketitle              
\begin{abstract}
In recent years, with the advancement of frontier AI, we have observed certain dynamics in open-sourcing and closed-sourcing decisions. We propose a game-theoretic model to analyze these dynamics in the current landscape of the AI race.  Our model builds on an R\&D race framework under a winner-takes-all setting, and it accounts for the cases where the players’ actions can be either discrete or continuous (i.e., partial open-sourcing, such as open weights). 
We show that determining the existence of a discrete pure non-trivial Nash equilibrium is NP-hard in general but that 
we can transform the discrete Nash existence computation into a MIP (Mixed-Integer Programming) problem, making it tractable for small instances using a standard MIP solver. Next, we show the existence and tractability of pure Nash equilibria in the continuous version of our problem, leveraging standard convex analysis results, and constructing an equivalent MIP formulation. Throughout this work, we leverage both our main technical results as well as surrounding technical analysis, to derive socially relevant insights that we believe can serve both to understand already existing decisions and dynamics and to potentially inform new policies.

\end{abstract}
\section{Introduction}\label{sec:introduction}

The question of open-sourcing AI has become increasingly contentious within the AI community in recent years. Pursuing one of the potentially most transformative technologies of the modern era requires critical attention to safety and transparency. One of the major questions that arises is whether AI models should be open or closed source.
Frequently cited arguments in favor of open source models emphasize their potential to promote democratization, deter monopolization, and support innovation~\citep{eiras2024near}. Furthermore, \citet{wheelerOpenSourceSecurity} argues that the security of software systems can be improved by open-sourcing code, as this allows external auditing by the broader public and consequently increases vulnerability detection.
In contrast, the arguments against it emphasize that open-sourcing powerful AI models increases the risk of misuse, and therefore it leads to concerns about potential safety risks and societal harms~\citep{eiras2024near}. 

While the debate between open source and closed source AI has received significant attention, most of the current work has focused on answering questions surrounding the implications of open sourcing and consequently, whether open source is ethically the most sound answer.

In our work, we shift the focus from asking \textit{whether} open-sourcing should occur and \textit{what} its consequences might be, to instead providing a mathematical model that allows us to ask \textit{when} and \textit{why} open-sourcing would occur in the first place—within the context of the current geopolitical and technological AI race.

We argue that, besides being an interesting question to answer, which would help us better understand the current state of the AI race and its near future, we believe that approaching this problem primarily from a strategic angle will inform more moral and effective policies.
Specifically, different regulatory frameworks, such as the EU AI Act~\citep{EUAiAct2024}, have taken initiatives to promote open-sourcing strategies among AI development teams by providing certain exemptions to the Act, which ultimately reduce the cost of compliance and could be used towards R$\&$D efforts. While these exemptions are a great step toward more open AI development, they do not necessarily take into account the positions of AI development teams in the AI race and the dynamics between them. Therefore, understanding the current state of the AI race would allow for more flexible and nuanced regulations, leading to more effective policies.

While we do not deny the importance of economic and political specifics, we argue that a certain level of abstraction is not only useful but necessary as a first step toward understanding the dynamic nature of the AI race. Our model centers on the relationships between key actors -- whether companies or countries, which we collectively refer to as AI development teams -- striving to advance AI and ultimately take the lead in the AI race. We introduce this abstraction to ensure both that our model generalizes well across different levels of AI race dynamics and that it remains relevant to similar technological races, as also argued by~\citet{armstrong2016racing}.

\textbf{Our assumptions:}

While approaching the AI race involves a plurality of strategic perspectives, we focus on scientific progress and the advancement of AI models. First, we base this decision on available empirical evidence, such as current benchmarks, research publications, and public model reports~\citep{lmsys2025chatbotarena, openai2023gpt4, anthropic2025claude4, polyak2024moviegen}. We argue that benchmarks such as Chatbot Arena track progress across different technical performance metrics, including logical and mathematical reasoning, coding capabilities, and multimodal capabilities, and are widely used by communities in academia and industry to compare models and their overall performance. Furthermore, media coverage frequently uses these results when reporting about the AI race, and such these benchmarks are broadly recognized as measures of progress~\citep{economist2024ai, techradar2025claude_benchmark}. Second, we refer to existing AI race modeling work that adopts a similar strategic focus~\citep{young2025s, armstrong2016racing, cimpeanu2022artificial}.

The central assumption in our model is that by open-sourcing their AI models, AI development teams can benefit from contributions from the broader AI community. These contributions may range from direct contributions, such as bug fixes, potential model improvements, and different technical extensions, to indirect contributions, such as reputational gains, which can lead to improved recruitment opportunities and ultimately greater scientific progress~\citep{bonaccorsi2003open}. 

However, while the team that opens its model gains certain advantages, competing AI development teams in the race may also benefit by having access to their competitors’ models, both for the revealed knowledge and for building upon them. Given this tension in the AI race, we seek to understand when and why open-sourcing an AI model would be a rational choice for an AI development team.

\textbf{Our contributions:}

We propose a game-theoretic model to analyze when and why AI development teams choose to open-source their models as a strategic decision. To the best of our knowledge, this is the first model of the AI race in which players' actions explicitly involve committing to open-source or closed-source strategies.

We analyze a general setting of 
$n$-player games, allowing both continuous and discrete action spaces, which makes our model general and widely applicable. In the discrete actions setting, we establish the existence of Nash equilibria using well-known results, derive necessary and sufficient conditions for the existence of pure Nash equilibria, analyze their computational complexity, and propose tractable solution methods—an important step toward practical applications of our framework.
In the continuous actions setting, we establish the existence of Pure Nash equilibria results and propose a tractable solution method for finding them.

 Finally, throughout our work we derive additional technical analysis motivated both by recent events in the AI race landscape and by broader prosocial questions that may inform future policy design. Specifically, in our work we focus on questions related to Pure Nash equilibrium, since we believe its interpretation is appropriate for real life deterministic and stable decisions in the AI race.


\section{Related Work}\label{sec:related_work}

\textbf{AI Race Related Work}

This work proposes a game-theoretic model of the AI race, aiming to better understand the dynamics of a race where the possible actions are strictly defined as open-source or closed-source. While several prior studies have examined AI race dynamics, most have framed actions primarily in terms of safety distinguishing between “safe” and “unsafe” behaviors.
One of the seminal works on the AI race, \citet{armstrong2016racing} analyze it through the lens of Nash equilibria, where players’ payoffs represent a tradeoff between capabilities and safety. Players are incentivized to maximize their expected utility, but taking overly risky actions can result in zero utility. Along similar lines, \citet{cimpeanu2022artificial} examine the AI race with actions classified as safe or unsafe under comparable assumptions. \citet{pereira2020regulate} and \citet{young2025s} both concentrate on regulation: \citet{pereira2020regulate} discuss when regulatory intervention is necessary, while \citet{young2025s} models the AI race as an n-player game with resource sharing and emphasizes regulatory implementation to ensure safe outcomes. Young’s sharing option~\citep{young2025s} is comparable to our open-source action, although their primary emphasis on regulation differs from our focus. 
Finally, \citet{askell2019role} conceptually analyze responsible AI development as a collective action problem and highlight factors that could foster cooperation among AI companies towards that goal.

\textbf{Open Source Related Work}

While most of the related work on the AI race does not explicitly frame actions in terms of open-sourcing, a separate body of literature focuses on the open-sourcing of AI models. Specifically, several recent works examine the trade-offs between open-source and closed-source approaches. Both \citet{eiras2024near} and \citet{seger2023open} outline benefits and risks. While \citet{eiras2024near} argue that open-sourcing is the appropriate choice during the early-to-mid stages of model development and technology adoption, \citet{seger2023open} consider alternative strategies for sharing. While this line of work could be described as prescriptive—offering guidance on what the right choice might be—our more descriptive approach aims to clarify when and why such a choice would occur.

\textbf{R\&D Race and Winner Takes it all Related Work}

Work on the AI race taking the form of a perpetual R\&D race was first discussed in~\citep{askell2019role}. 
Work on patent races and competitive dynamics in innovation has been extensively studied throughout the years, starting with classical works such as~\citep{KamienSchwartz1972} and~\citep{Reinganum1981}, to more recent experimental analysis of the horizons of R\&D races under a certain level of parameter uncertainty~\citep{breitmoser2010understanding}, something that could be useful for the extension of our model.

\section{Model}\label{sec:model}

Consider the AI race between $n$ players, in which each player can choose between open-sourcing or closed-sourcing their model. We propose a discrete and continuous model of the AI race, in which players aim to maximize their scientific progress. Specifically, we analyze the dynamics of games where players can simply decide to fully open source or closed source
their models $a_i\in\{0,1\}$, where 0 action represents closed source, and 1 action represents open source (discrete model). Then we extend our discrete model to a continuous model where the players can choose to partially open source their model, i.e. $a_i\in\left[0,1\right]$ (continuous model).

Companies have a history of open-sourcing their models~\citep{touvron2023llama, rombach2022high, biderman2023pythia}. There are several benefits to doing so~\citep{bonaccorsi2003open}. Specifically, we denote $\delta_i$ as the scientific progress player $i$ gains from the community by going open source, and since other players benefit from player $i$ going open source, we denote $\Delta_{ij}$ as the gains player $j$ gains from player $i$ open sourcing their model. We assume that these gains are independent, which helps us avoid introducing an exponential number of variables.

 We also note that our model captures the natural assumption that not all players are at the same position in the AI race, and we denote their initial position (in terms of scientific progress) with the variable $s_i$.

Their new position would be: 

\[
\mu_i(a_1, \dots, a_k) = \delta_i a_i + \sum_{\substack{j = 1 \\ j \ne i}}^k \Delta_{ji} a_j + s_i + \Delta_i
\]

Where $\Delta_i$ represents the internal scientific progress player $i$ achieves independently of external collaboration or community support.
For the clarity of results we introduce new variable $d_i$, variable $d_i$ is simply the baseline position of player $i$ in the race after internal progress without external effects.
We rewrite their new position as:
\[
\mu_i(a_1, \dots, a_k) = \delta_i a_i + \sum_{\substack{j = 1 \\ j \ne i}}^k \Delta_{ji} a_j + d_i
\]

According to our setting R\&D race in which winner takes all players are not necessarily optimizing just for their position but their position to the best competitor’s position.

\[
u_i(a_1, \dots, a_k) 
= \mu_i(a_1, \dots, a_k) 
- \max_{j \neq i} \, \mu_j(a_1, \dots, a_k).
\]

\textbf{Matrix form.}
We can rewrite our model in matrix-vector form. Let $\mathcal{D}\in\mathbb{R}^{k\times k}$ be defined by
\[
\mathcal{D}_{ij}=
\begin{cases}
\delta_i & i=j,\\
\Delta_{ji} & i\neq j,
\end{cases}
\]

and let $\mathbf{a}=(a_1,\ldots,a_k)^\top$ and 
$\mathbf{d}=(d_1,\ldots,d_k)^\top$.

\medskip
\noindent
Then the progress profile is
\[
\mu(\mathbf a)=\mathcal D\mathbf a+\mathbf d \in \mathbb R^k,
\]
whose $i$-th component $\mu_i(\mathbf a)$ represents the progress of player $i$.
\medskip
\noindent

Player utilities are defined as
\[
u_i(\mathbf a)=\mu_i(\mathbf a)-\max_{j\neq i}\mu_j(\mathbf a)
=
e_i^\top\mu(\mathbf a)-\max_{j\neq i}e_j^\top\mu(\mathbf a),
\]
where $e_i$ is the $i$-th unit vector.

\section{Nash Equilibrium - Discrete Actions}\label{sec:NE_discrete_actions}
As we have mentioned in previous sections, most of the work on open source has been \textit{prescriptive}~\citep{eiras2024near} in its nature, while we aim for our work to be more \textit{descriptive}. We emphasize the importance of \textit{descriptive} as a first step towards more implementable potential policies. In that pursuit, we believe it is important to answer several crucial questions addressed in this section. First, while our proposed model's potential choices (actions) are defined as open source and closed source, we recognize that the true importance of those potential actions lies in their stability.
In this regard we are interested in answering when does our game have Nash equilibrium. We emphasize that the answer to this question is crucial for understanding stable long-term strategies 
and therefore implementing robust policies. But even more than just understanding understanding of the existence of Nash equilibrium, we are interested in the existence of Pure Nash equilibrium. We are interested in answering what the necessary and sufficient conditions are for AI development teams to have a deterministic and stable strategy, and in this section we answer questions such as what are necessary and sufficient for the existence of Pure Nash equilibrium.

\subsection{The Existence of Pure Nash Equilibrium - Discrete Actions}
\begin{lemma}
[Pure Nash Equilibrium Existence]\label{thm:thm1}
Consider the game defined in Section~\ref{sec:model}. 
A pure strategy profile $\mathbf{a}\in\{0,1\}^k$ is a Nash equilibrium if and only if for every player $i$,
\[
\max_{j \ne i} \mu_j(\mathbf{a})
- \max_{j \ne i} \mu_j(\mathbf{a}^{(i)})
+ (-1)^{a_i}\delta_i \le 0,
\]
where $\mathbf{a}^{(i)}=(a_1,\dots,1-a_i,\dots,a_k)$.
\end{lemma}

\begin{proof}
    Pure strategy profile \( \mathbf{a} = (a_1, \dots, a_k)\) is a Nash equilibrium if and only if for every player \(i\),
\[u_i(\textbf{a})
 \geq u_i(\mathbf{a}^{(i)})
\]

Then our Nash condition for every player $i$:
\[u_i(\mathbf{a})
  - u_i(\mathbf{a}^{(i)}) \geq 0 
\]

becomes  
\begin{align*}
&= u_i(\mathbf{a}) - u_i(\mathbf{a}^{(i)})\\
&= - \max_{\substack{j = 1 \\ j \ne i}}^{k} \mu_j(\mathbf{a}^{(i)}) +  \max_{\substack{j = 1 \\ j \ne i}}^{k} \mu_j(\mathbf{a}) + \delta_i (1 - a_i) - \delta_i a_i\\
&= - \max_{\substack{j = 1 \\ j \ne i}}^{k} \mu_j(\mathbf{a}^{(i)}) +  \max_{\substack{j = 1 \\ j \ne i}}^{k} \mu_j(\mathbf{a}) + \delta_i (1 - 2 a_i)\\
&= - \max_{\substack{j = 1 \\ j \ne i}}^{k} \mu_j(\mathbf{a}^{(i)}) + \max_{\substack{j = 1 \\ j \ne i}}^{k} \mu_j(\mathbf{a}) + (-1)^{a_i} \cdot \delta_i
\end{align*}
\end{proof}

The result obtained in \autoref{thm:thm1} aligns with the intuition from our game formulation: given that the strategies of all other players are fixed, an AI development team will choose to open source their code if and only if the benefit they receive from the community exceeds the difference between the benefit obtained by the best competitor when the team open sources their research and the benefit obtained by the best competitor when the team keeps the results of their work closed, given that the identity of the best competitor may change across these two scenarios.

While we have stated above the interpretation of the results of \autoref{thm:thm1}, we observe that the possibility of having different best competitors makes these results highly sensitive to the accuracy of all parameters in all utility functions of players.
In order to make our interpretation more robust and therefore more applicable for real-world examples in complex (often messy) systems we outline bellow  \autoref{cor:cor1} sufficient conditions for the existence of Pure Nash equilibrium. 

\begin{corollary}[Sufficient Condition for Pure Nash Equilibrium]\label{cor:cor1}
A pure strategy profile $(a_1, \dots, a_k)$ is a Nash equilibrium if, for every player $i$:
\[
\begin{cases}
\max_j \Delta_{ji} < \delta_i & \text{for } a_i = 1,\\[2mm]
\delta_i < \min_j \Delta_{ji} & \text{for } a_i = 0.
\end{cases}
\]
\end{corollary}

\textbf{Deviation from Close Source Strategy}
Furthermore, we use our model to address an important and, given recent events, highly relevant question. 
Specifically, in \autoref{cor:cor2} we examine the following: 
if all AI development teams choose to keep their models closed source, 
under what conditions would one of them have an incentive to deviate and open source their results? This question holds particular relevance for future policymakers, who may seek mechanisms to encourage some AI development teams to open source their models, particularly in situations where all other teams have opted to keep their research closed.

\begin{corollary}\label{cor:cor2}
If all players choose the closed-source action, a player $i$ will have an incentive to deviate and the profile will not be a Nash equilibrium if:
\[
\boxed{
\max_{j \ne i} d_j
-
\max_{j \ne i} (d_j + \Delta_{ij})
 \geq -\delta_i.
}
\]
\end{corollary}

The result above has a natural practical interpretation. In many realistic settings, top players are already significantly ahead in the race, reflected by large values of $d_j$. For late or weaker players, it is reasonable to assume that their open-sourcing decision would have only a limited impact on the leaders, i.e., $\Delta_{ij}$ is relatively small. At the same time, the direct effect of open-sourcing on the weaker player—captured by $\delta_i$—may be strictly positive due to community contributions and reputational gains. Under such conditions, the inequality is more likely to hold, implying that players ranked lower in the AI race have an incentive to deviate and open-source their results. Intuitively, when a player is sufficiently behind, the potential self-benefits of openness can outweigh the risk of further strengthening already leading players.

\textbf{Social Welfare Maximization at Pure Nash Equilibrium}
Finally, beyond analyzing the conditions under which AI teams would deviate from closed-source strategies—a result we believe is relevant for potential regulatory interventions aimed at incentivizing open source—we ask a broader question: \textit{when does the choice to open-source, or more broadly when does strategic behavior align with socially desirable outcomes?} 

Specifically, we define a socially desirable outcome as one that maximizes total scientific progress. We argue that this objective is well motivated for two reasons. First, defining social welfare in terms of sum of utilities is well established in the game theory literature~\citep{harsanyi1955cardinal, bergson1938reformulation}. Second, the importance of promoting and potentially advancing scientific progress has been widely emphasized in the AI governance and safety literature~\citep{dafoe2018ai}.

\begin{corollary}
Let social welfare be defined as total scientific progress
\[
W(\mathbf{a})=\sum_{i=1}^k \mu_i(a_1,\dots,a_k),
\]

Then a pure Nash equilibrium maximizes social welfare if and only if, for every player $i$,
\[
a_i^{\mathrm{NE}}=
\begin{cases}
1, & \text{if } 
\delta_i \ge 
\max_{j\neq i}\mu_j(a_i=1,a_{-i}^*)-
\max_{j\neq i}\mu_j(a_i=0,a_{-i}^*),\\
0, & \text{otherwise}.
\end{cases}
\]

In particular, a sufficient condition for this alignment is
\[
a_i^{\mathrm{NE}}=
\begin{cases}
1, & \text{if } \delta_i \ge \max_{j\neq i}\Delta_{ji},\\
0, & \text{if } \delta_i \le \max_{j\neq i}\Delta_{ji}.
\end{cases}
\]
\end{corollary}


\subsection{Complexity of Finding Pure Nash Equilibrium - Discrete Actions}

In the previous subsection, we presented results regarding sufficient and necessary conditions for the existence of (Pure) Nash equilibrium and provided their interpretation. In this section we will address the question of the computational complexity of finding such equilibria. 
While the solutions to Pure Nash Equilibrium can be computed for several players, we argue that the following section besides game theoretical contribution is important for scalability of our model. In this section we prove that determining existence of non-trivial  Pure Nash Equilibrium is NP-hard. However, we also show that our problem can be formulated as Mixed-Integer Program (MIP), making it tractable using MIP solvers.

\begin{theorem}
In a game defined in \autoref{sec:model}, where each player chooses open source or closed source action (binary action) determining the existence of a non-trivial pure Nash equilibrium is NP-hard, where the trivial equilibrium is the profile in which all players choose the closed-source action.
\end{theorem}
\begin{proof}

To prove that determining the existence of a non-trivial pure Nash equilibrium is NP-hard, we construct a polynomial-time reduction from the well-known NP-complete problem 3-SAT to our formulation of the Pure Nash Equilibrium problem. In this proof we will show that  if and only if the solution that satisfies 3-SAT problem exists, Pure Nash Equilibrium exists in our game.

\textbf{3-SAT Problem}
The 3-SAT problem is a well-known NP-hard problem that is a special case of the general Boolean satisfiability problem. The 3-SAT problem has a formula conjunctive normal form (CNF) and each clause has exactly three literals. Without loss of generality let us analyze the following 3-SAT problem:

\[
\Phi = C_1 \land C_2 \land \dots \land C_m
\]

where each clause \( C_i \) is taking the form:

\[
C_i = (l_{i1} \lor l_{i2} \lor l_{i3})
\]
\begin{itemize}
    \item Each literal is either a variable \( x_t \) or its negation \( \neg x_t \), for some \( 1 \leq t \leq k \).
    \item \( k \) Boolean variables \( x_1, x_2, \dots, x_k \).
\end{itemize}

The goal is to determine whether there exists an assignment $\varphi$ to the variables \( x_1, x_2, \dots, x_k \) such that the formula \( \Phi(\varphi) \) is true.
Finding an instance $\varphi$ for which \( \Phi(\varphi) \) is true known to be NP hard, while verifying if an instance $\varphi$ satisfies formula $\Phi$ is possible to be done in polynomial time.










\textbf{Game Construction:}
Now we want to construct game $G(a(\varphi))$ that we will be able to formulate with polynomial translation from a given 3-SAT problem.
Consider the game that has $k + m + 1$ players in total, with each one of them having action $a_i$ where $i = 1 \dots k+ m + 1$. We name these players differently according to their role and their utility, although their utilities is still defined as in \autoref{sec:model}.
We classify those players according to their utilities and their roles in the following manner: 
\begin{itemize}
    \item $k$ Variable players 
    \item $m$ Clause players
    \item $1$  Switch player
\end{itemize}

\textbf{Variable Players:}
Let's say that any player $i$ where $i \in \{1, \dots k\}$ is variable player. There actions are $a_1, \dots a_{k}$. Action of each Variable player corresponds to variable $x_k$ in original 3-SAT problem ($\Phi$). However, please note that their position utility $\mu$ does not depend on their action because for each variable player $i$ we define their $\mu$ ranking utility as:
\[
    \mu_{i}(\mathbf{a}) = 0
\]
Then their overall utility $u$ is:
\[
    u_{i}(\mathbf{a}) = \mu_{i}(\mathbf{a}) - \max_{j \ne i} \mu_j(\mathbf{a}) = - \max_{j \ne i} \mu_j(\mathbf{a})
\]

\textbf{Clause Players:}
Now we define clause player $i$ where $i \in \{k+1, \dots m+k\}$. Each clause player's utility captures the value of one clause in original SAT problem $\Phi$. So for each clause

Each clause player $j$ defines a linear function based on its clause $C_j = \ell_{j1} \lor \ell_{j2} \lor \ell_{j3}$, where each literal $\ell_{jk}$ corresponds to a variable index $i_{jk} \in \{1, \dots, n\}$:
\begin{align*}
    f_{C_j}(a) = \sum_{k=1}^{3} \begin{cases}
        a_{i_{jk}} & \text{if } \ell_{jk} = x_{i_{jk}} \\
        1 - a_{i_{jk}} & \text{if } \ell_{jk} = \lnot x_{i_{jk}}
    \end{cases}
\end{align*}

Let $M \gg 0$ and $\alpha = 1$ and $\beta \ll \alpha$ . Now we define clause utility as:
\[
    \mu_j(a) = M - \alpha \cdot f_{C_j}(a) + \beta(a_1 + \dots + a_k)
\]
So:
\begin{itemize}
  \item If clause is \textbf{satisfied}: $f_{C_j}(a) \ge 1$ $\Rightarrow$ $\mu_j(a) < M - 0.5$
  \item If clause is \textbf{violated}: $f_{C_j}(a) = 0$ $\Rightarrow$ $\mu_j(a) \geq M$
\end{itemize}

\textbf{Switch Player:} Finally, we define switch player $i = m+k+1$ in the following manner:

\[\mu_{m+k+1}(a) = M - 0.5\]

\noindent\textbf{Direction (\(\Rightarrow\)):}
{\bfseries\boldmath
\[
\Phi(\varphi) = \text{True} \;\Rightarrow\; \mathbf{a(\phi)} \text{ is a pure Nash equilibrium.}
\]
}

If $\Phi(\varphi) = \text{True}$ it means that every clause $C_j$ is satisfied for instance $\varphi$. 

This means that the ranking utilities $\mu$ have following values:
\begin{itemize}[label=$\circ$, left=0pt]
    \item $\mu_i(\mathbf{a}) = 0 \quad \forall i \in \{1,\dots,k\}$ (Variable Players)
    \item $\mu_i(\mathbf{a}) < M - 0.5\quad \forall i \in \{k+1,\dots,k+m\}$ (Clause Players)
    \item $\mu_i(\mathbf{a}) = M - 0.5 \quad i = k+m+1$ (Switch Player)
\end{itemize}

Since the utilities $u_i(\mathbf{a})$ for Clause Players and Switch Player does not depend on their action $a_i$ we have that the following Pure Nash condition: 

\[u_i(\mathbf{a}) - u_i(\mathbf{a}^{(i)}) \geq 0\]
always holds.

In order to show that it is Pure Nash Equilibrium we want to show that above condition holds for Variable Players too.
As stated in our construction we know that for Variable Players:
\[u_i(\mathbf{a}) =  - \max_{j \ne i} \mu_j(\mathbf{a}) = - M + 0.5\]
Now notice since potential values for $u_i(\mathbf{a}) = - M + 0.5$ and $u_i(\mathbf{a}) = - M$ (when one of the clauses is violated) we showed that player $i$ truly cannot improve their value in any way and therefore the \[u_i(\mathbf{a}) - u_i(\mathbf{a}^{(i)}) \geq 0\] holds for Variable Players too, which finalizes this part of proof since we have just shown that Nash Equilbirum condition holds for all players.

\noindent\textbf{Direction (\(\Leftarrow\)):}
\textbf{\boldmath
\[
\Phi(\varphi) = \text{True} \;\Leftarrow\; \mathbf{\phi} \text{ is a pure Nash equilibrium.}
\]
}

Proving that 
\[
\varphi \text{ is a Pure Nash equilibrium} \quad \implies \quad \Phi(\varphi) = \text{True}
\]
is equivalent to showing that 
\[
\Phi(\varphi) \neq \text{True} \quad \implies \quad \varphi \text{ is not a Pure Nash equilibrium}.
\]
So, lets assume that $\Phi(\varphi) \neq \text{True}$. If $\Phi(\varphi)$ is not true that means that there exists at least one clause $C_j$ that is not satisfied.

Notice that if $C_j$ is not satisfied, $f_{C_j}(a) = 0$ which implies \[
    \mu_j(a) = M + \beta(a_1 + \dots + a_k)
\].

This means that the ranking utilities $\mu$ have the following values:

\begin{itemize}[label=$\circ$, left=0pt]
    \item $\mu_i(\mathbf{a}) = 0 \quad \forall i \in \{1,\dots,k\}$ (Variable Players)
    
    \item $\forall i \in \{k+1,\dots,k+m\}$ (Clause Players)
    \begin{itemize}[label=--, left=1.5em]
        \item Case 1: $\mu_i(\mathbf{a}) < M$ for $i \in \{k+1,\dots,k+m\} \setminus S$  
        \item Case 2: $\mu_i(\mathbf{a}) = M$ for $i \in S \subseteq \{k+1,\dots,k+m\}$, where $S \neq \emptyset$
    \end{itemize}

    \item $\mu_i(\mathbf{a}) = M - 0.5 \quad i = k+m+1$ (Switch Player)
\end{itemize}

Similarly, as in first part of this proof we are only interested in analyzing if Nash Equilibrium condition holds for Variable Players (since for other players it always holds).

In order to show that it is not Pure Nash Equilibrium we want to show that above condition does not for some Variable Player.
As stated in our construction we know that for Variable Players:
\[u_i(\mathbf{a}) =  - \max_{j \ne i} \mu_j(\mathbf{a}) = - M - \beta(a_1 + \dots + a_k)\]

Now lets analyze two cases: 

\underline{\text{Case 1:}}
$(a_1, \dots a_i \dots, a_k) = (0, \dots 0)$
If $\mathbf{a} = (0, \dots 0)$, we have that the utility of Variable players is:
\[u_i(\mathbf{a}) = - M\]

In which case $u_i(\mathbf{a}^{(i)})$ can be either $-M$ either $-M + 0.5$ (if the change in one action lead to all clauses being satisfiable). In the first case $\mathbf{a}$ was not Pure Nash Equilibrium to start, in the second case it is Pure Nash Equiliburm, but it is not non-trivial Pure Nash Equilibirum.

\underline{\text{Case 2:}}
For any $(a_1, \dots a_i \dots, a_k)$
where $(a_1, \dots a_i \dots, a_k)\neq (0, \dots 0)$.

In this case it means that there exists $i$ such that $i\neq0$.

Given that 

\[u_i(\mathbf{a}) =  - \max_{j \ne i} \mu_j(\mathbf{a}) = - M - \beta(a_1 + \dots + a_k)\]

we can see that $u_i(\mathbf{a}^{(i)}) \geq u_i(\mathbf{a})$ since by flipping $i$, $u_i(\mathbf{a}^{(i)}) = u_i(\mathbf{a}) + \beta$ or $u_i(\mathbf{a}^{(i)}) = -M$ (in case all clauses became satisfied by this change). Either way, we can see that therefore $\mathbf{a}$ does not satisfy solution for Pure Nash Equilibrium.

\end{proof}

\textbf{Pure Nash Equilibrium as MIP}

While we have showed in subsection above that determining the existence of a non-trivial solution for Pure Nash Equilibrium is NP hard, in the theorem bellow it can be reformulated as Mixed-Integer Program.

\begin{theorem}[MIP Characterization of Pure Nash Equilibria]
    A pure Nash equilibrium in this game can be formulated as a Mixed-Integer Program (MIP).
\end{theorem}

\begin{proof}

We want to show that the condition for Pure Nash Equilibrium can be rewritten as MIP (Mixed-Integer Program) and therefore we can use MIP solvers to solve it.

\begin{small}
\begin{center}
\fbox{%
\parbox{0.97\columnwidth}{%
\textbf{MIP formulation for Pure Nash Equilibrium}

\textbf{Variables:}
\[
\begin{aligned}
&a_i \in \{0,1\} && \text{($k$ player actions)}\\
&z_{i}, \tilde z_i  \in \mathbb{R} && \text{($2k$ cont.)}\\
&y_{i,j}, \tilde y_{i,j} \in \{0,1\} && \text{($2k$ indicators)}
\end{aligned}
\]

\textbf{Constraints:}

\[
\begin{aligned}
&z_{i} \ge \mu_j(\mathbf{a}), &&\forall i,j\ne i \\
&z_{i} \le \mu_j(\mathbf{a}) + (1-y_{i,j})M, &&\forall i,j\ne i\\
&\sum_{j\ne i} y_{i,j} = 1, &&\forall i
\end{aligned}
\]

\[
\begin{aligned}
&\tilde z_{i} \ge  \mu_j(\mathbf{a}^{(i)}), &&\forall i,j\ne i \\
&\tilde z_{i} \le \mu_j(\mathbf{a}^{(i)}) + (1-\tilde y_{i,j})M, &&\forall i,j\ne i \\
&\sum_{j\ne i} \tilde y_{i,j} = 1, &&\forall i
\end{aligned}
\]

\[
z_{i} - \tilde z_{i} \ge (2a_i - 1)\delta_i,
\quad \forall i,\mathbf{a}
\tag{N}
\]

\textbf{Objective:} $\min 0$ (feasibility)
}%
}
\end{center}
\end{small}


Above we stated our problem as MIP formulation. We have stated over which variable we are optimizing. Our problem is of feasibility nature where we are trying to find variables that satisfy the constraint and therefore there is no objective we are optimizing over, since Pure Nash Equilibrium condition is perfectly captured in our constraints.

Specifically,

\begin{equation}
\begin{aligned}
&z_{i} \ge \mu_j(\mathbf{a}), &&\forall i,j\ne i \\
&z_{i} \le \mu_j(\mathbf{a}) + (1-y_{i,j})M, &&\forall i,j\ne i\\
&\sum_{j\ne i} y_{i,j} = 1, &&\forall i
\end{aligned}
\end{equation}

This construction ensures that $z_{i} = \max_{j \ne i} \mu_j(\mathbf{a})$. Let us examine these conditions for a fixed player $i$. First, note that the constraint
\[
\sum_{j \ne i} y_{i,j} = 1,
\]
together with $y_{i,j} \in \{0,1\}$, guarantees that exactly one $y_{i,j^*} = 1$ while all other $y_{i,j} = 0$. Then, we observe that
\begin{equation}
\begin{aligned}
z_{i} &\ge \mu_j(\mathbf{a}), &&\forall j\ne i \\
z_{i} &\le \mu_j(\mathbf{a}) + (1-y_{i,j})M, &&\forall j\ne i
\end{aligned}
\end{equation}
forces $z_i = \mu_{j^*}(\mathbf{a})$. Now, suppose for contradiction that $z_i$ is not the maximum, i.e., there exists some $t \ne j^*$ such that
\[
\mu_t(\mathbf{a}) = \max_{j \ne i} \mu_j(\mathbf{a}) > \mu_{j^*}(\mathbf{a}).
\]
By construction, $y_{i,t} = 0$ and the lower bound constraint $z_i \ge \mu_t(\mathbf{a})$ still holds. But then we would have
\[
z_i = \mu_{j^*}(\mathbf{a}) < \mu_t(\mathbf{a}) \le z_i,
\]
which is a contradiction. Therefore, the assumption is false, and it follows that
\[
z_i = \max_{j \ne i} \mu_j(\mathbf{a}).
\]

Similarly, we have $\tilde z_i = \max_{j \ne i} \mu_j(\mathbf{a}^{(i)})$ for each player $i$ and profile $\mathbf{a}$. These two conditions with \[
z_{i} - \tilde z_{i} \ge (2a_i - 1)\delta_i,
\quad \forall i,\mathbf{a}
\tag{N}
\] perfectly capture Pure Nash Equilbirum condition.

\end{proof}
\section{Nash Equilibrium - Continuous Actions}\label{sec:NE_continuous_actions}
In the previous section, we presented results for the case in which players’ actions are discrete. In this section, we analyze cases where players are allowed to partially open-source their models (i.e., open weights).

\subsection{The Existence of Pure Nash Equilibrium - Continuous Actions}

\begin{theorem}[Existence of Pure Nash Equilibrium]
Consider a $k$-player game where each player $i \in \{1, \dots, k\}$ chooses an action $a_i \in [0,1]$, and the utility function for player $i$ is given by
\[
u_i(a_1, \dots, a_k) = \mu_i(a_1, \dots, a_k) - \max_{\substack{j=1 \\ j \ne i}}^{k} \mu_j(a_1, \dots, a_k),
\]
where
\[
\mu_i(a_1, \dots, a_k) = \delta_i a_i + \sum_{\substack{j=1 \\ j \ne i}}^{k} \Delta_{ji} a_j + d_i,
\]
with constants $\delta_i, \Delta_{ji}, d_i \in \mathbb{R}$.

Then, a Pure Nash equilibrium exists.
\end{theorem}
 \begin{proof}
 
Each action set $[0,1]^k$ is nonempty, compact, and convex. The utility function $u_i(\mathbf{a})$ is continuous since it is the difference of linear function $\mu(\mathbf{a})$ (continuous function) and a maximum over finitely many continuous functions $\max_{\substack{j=1 \\ j \ne i}}^{k} \mu_j(\mathbf{a})$ (continous function) . Furthermore, utility function $u_i(\mathbf{a})$ is quasi-concave in $a_i$. Therefore, by Glicksberg's generalization of Nash's theorem ~\citep{glicksberg1952}, a Pure Nash equilibrium exists.
\end{proof}

\subsection{Complexity of Finding Pure Nash Equilibrium — Continuous Actions}

For fixed $a_{-i}$, player $i$'s utility 
  $u_i(a_i,a_{-i})=\mu_i(a)-\max_{j\neq i}\mu_j(a)$ is concave and piecewise-linear in $a_i$. By concavity, we know that in order for an action $a_i^*$ to be a best response one of the following slope conditions needs to hold:
(i) if $a_i^*\in(0,1)$ the slope is $0$,
(ii) if $a_i^*=0$ the slope is $\le 0$,
and (iii) if $a_i^*=1$ the slope is $\ge 0$.
We reformulate these conditions bellow, using linear constraints together with a linearization of the max operator.

\begin{theorem}[MIP Characterization of Pure Nash Equilibria]
\label{thm:milp_continuous}
Consider the continuous game introduced in Section~\ref{sec:model}. 
A strategy profile $\mathbf{a}^* \in [0,1]^k$ is a pure Nash equilibrium if and only if there exist auxiliary variables
\[
m_i \in \mathbb{R}, \quad y_{ij} \in \{0,1\}, \quad \lambda_{ij} \ge 0, \quad 
b_i^0, b_i^1 \in \{0,1\}, \quad \alpha_i, \beta_i \ge 0,
\]
satisfying the constraint system \eqref{eq:mip-system}. In particular, the set of pure Nash equilibria coincides with the feasible region of this MIP.
\end{theorem}
\begin{subequations}\label{eq:mip-system}
\begin{align}
&\mu_i = \delta_i a_i + \sum_{j\neq i}\Delta_{ji} a_j + d_i, \\
&m_i \ge \mu_j, && \forall j\neq i, \\
&m_i \le \mu_j + M(1-y_{ij}), && \forall j\neq i, \\
&\sum_{j\neq i} y_{ij} = 1,
\qquad
\sum_{j\neq i}\lambda_{ij} = 1, \\
&0 \le \lambda_{ij} \le y_{ij}, 
\qquad
g_i = \sum_{j\neq i}\lambda_{ij}\Delta_{ij}, \\
&b_i^0 + b_i^1 \le 1,
\qquad
a_i \le 1 - b_i^0,
\qquad
a_i \ge b_i^1, \\
&0 \le a_i \le 1,
\qquad
\alpha_i \le M_\alpha b_i^0,
\qquad
\beta_i \le M_\beta b_i^1, \\
&\delta_i - g_i + \alpha_i - \beta_i  = 0.
\end{align}
\end{subequations}
\begin{proof}
Notice that, similarly as in discrete case proof we formulate $m_i$ and the appropriate constraints such that $m_i$ represents the best competitor's payoff. 


Player $i$'s utility is:
\[
u_i(a_i,a_{-i})=\mu_i(\mathbf a)-\max_{j\neq i}\mu_j(\mathbf a),
\qquad
\mu_i(\mathbf a)=\delta_i a_i+\sum_{j\neq i}\Delta_{ji}a_j+d_i.
\]
Hence
\[
\frac{\partial \mu_i(\mathbf a)}{\partial a_i}=\delta_i.
\]
Let
\[
g_i:=\sum_{j\neq i}\lambda_{ij}\Delta_{ij},
\]
which represents a subgradient of the max term with respect to $a_i$ for the player $i$.
Therefore the slope of the player's $i$ utility with respect to $a_i$ is:
\[
\frac{\partial u_i(\mathbf a)}{\partial a_i}=\delta_i-g_i.
\]

As previously discussed, since $u_i(\cdot,a_{-i})$ is concave and piecewise-linear in $a_i$, an action $a_i^*\in[0,1]$
is a best response if and only if
\[
\begin{cases}
\delta_i-g_i=0 & \text{if } a_i^*\in(0,1),\\[4pt]
\delta_i-g_i\le 0 & \text{if } a_i^*=0,\\[4pt]
\delta_i-g_i\ge 0 & \text{if } a_i^*=1.
\end{cases}
\]

Let's introduce $b_i^0,b_i^1\in\{0,1\}$ variables and stated constraints bellow:
\[
b_i^0+b_i^1\le 1,
\qquad
a_i\le 1-b_i^0,
\qquad
a_i\ge b_i^1,
\qquad
0\le a_i\le 1.
\]
Now notice that there are three possible cases values for pair of variables $(b_i^0, b_i^1)$:  
Thus $b_i^0=1, b_i^1=0$ forces $a_i=0$, $b_i^0=0, b_i^1=1$ forces $a_i=1$, and $b_i^0=b_i^1=0$ corresponds to an interior solution.

By introducing $\alpha_i,\beta_i\ge 0$ and introducing the following contraints:
\[
\delta_i-g_i+\alpha_i-\beta_i=0,
\qquad
\alpha_i\le M_\alpha b_i^0,
\qquad
\beta_i\le M_\beta b_i^1.
\]
We make sure that if $a_i\in(0,1)$ then $b_i^0=b_i^1=0$ implies $\alpha_i=\beta_i=0$ and therefore $\delta_i-g_i=0$.
If $a_i=0$ then $\beta_i=0$ and $\delta_i-g_i=-\alpha_i\le 0$.
If $a_i=1$ then $\alpha_i=0$ and $\delta_i-g_i=\beta_i\ge 0$, which was exactly to be proven.
\end{proof}




\section{Discussion}

To the best of our knowledge, we propose the first game-theoretic model that analyzes open-sourcing strategies in the AI race. We comment on the existence of Nash equilibria, characterize conditions under which (Pure) Nash equilibria exist—both in discrete and continuous action spaces—and analyze the complexity of finding such equilibria. In both cases, we propose tractable methods for computing Pure Nash equilibria. Beyond these modeling contributions and theoretical results, we provide interpretations throughout our work, including mathematical reasoning that explains the latest events and may inform new regulatory policies.

We hope that our work constitutes a first step toward understanding open source as a rational and verifiable choice within the framework of game theory. In that spirit, many open questions remain to be addressed in future research.

Possible technical extensions of our work include strengthening our hardness results, for instance by extending our proof to show that determining existence of any pure Nash equilibrium is NP-hard, as well as relaxing certain modeling assumptions, such as the independence of gains across players. Besides potential technical contributions, this work opens the door to a number of conceptual contributions, such as extending the model to a Bayesian version of our game-theoretic model. Specifically, if players in the AI race do not have full knowledge of the benefits they or their competitors might obtain from open-sourcing, it would be important to analyze what their optimal strategies would be. Furthermore, the dynamics in the Bayesian version of our model could be affected by public statements that come with certain levels of commitment, which could change the overall dynamics and would therefore present an important setting to consider.

Finally, our model may provide a mathematical foundation for future policy design that could vary from policies such as different incentive-based mechanisms, which could affect gains in our model, to potential mandatory disclosure requirements, which could influence the information available about competitors.

\section{Acknowledgment}
A. Mladenovic is supported by an IVADO fund under the Canada First Research Excellence Fund grant to develop robust, reasoning, and responsible artificial intelligence, FRQNT fellowship and Mila EDI Excellence scholarship.
%
%
%
%
%
%





\bibliographystyle{splncs04}
\bibliography{refs}

\end{document}